\documentclass[amsthm]{elsart}

\usepackage{amsmath}
\usepackage{amsfonts} 
\usepackage{amssymb}

\usepackage{hyperref} 

\newtheorem{theorem}{Theorem}[section]
\newtheorem{lemma}[theorem]{Lemma}
\newtheorem{proposition}[theorem]{Proposition}
\newtheorem{corollary}[theorem]{Corollary}
\newtheorem{Definition}[theorem]{Definition}
\newtheorem{Example}[theorem]{Example}
\newtheorem{conjecture}[theorem]{Conjecture}

\newenvironment{example}{\begin{Example}}%
        {\hfill$\Box$\smallskip\end{Example}}

        {\hfill$\Box$\smallskip\end{trivlist}\vspace*{-.6cm}}

\def\Maple{\textsc{Maple}}
 \def\NN{{\mathbb N}} \def\PP{{\mathbb P}}
  \def\ZZ{{\mathbb Z}}
\def\ds{\displaystyle}  \def\bm{\boldsymbol}
\def\id{\text{Id}}

\begin{document}

\begin{frontmatter}
\title{Multihomogeneous Resultant Formulae for Systems with Scaled Support}

\author{Ioannis Z. Emiris}
 \address{
Department of Informatics and Telecommunications\\ 
National and Kapodistrian University of Athens, Greece}
\ead{lastname@di.uoa.gr}

\author{Angelos Mantzaflaris}
 \address{
GALAAD, INRIA M\'editerran\'ee\\
BP 93, 06902 Sophia Antipolis, France}
\ead{FirstName.LastName@inria.fr}

\begin{abstract}
  Constructive methods for matrices of multihomogeneous (or
  multigraded) resultants for unmixed systems have been studied by
  Weyman, Zelevinsky, Sturmfels, Dickenstein and Emiris.  We
  generalize these constructions to {\em mixed} systems, whose Newton
  polytopes are scaled copies of one polytope, thus taking a step
  towards systems with arbitrary supports.  First, we specify matrices
  whose determinant equals the resultant and characterize the systems
  that admit such formulae.  B\'ezout-type determinantal formulae do
  not exist, but we describe all possible Sylvester-type and hybrid
  formulae.  We establish tight bounds for all corresponding degree
  vectors, and specify domains that will surely contain such vectors;
  the latter are new even for the unmixed case.  Second, we make use
  of multiplication tables and strong duality theory to specify
  resultant matrices explicitly, for a general scaled system, thus
  including unmixed systems.  The encountered matrices are classified;
  these include a new type of Sylvester-type matrix as well as
  B\'ezout-type matrices, known as partial Bezoutians.  Our
  public-domain \Maple\ implementation includes efficient storage of
  complexes in memory, and construction of resultant matrices.
\end{abstract}

\end{frontmatter}

\newpage
\section{Introduction}

Resultants provide efficient ways for studying and solving polynomial
systems by means of their matrices. They are most efficiently expressed by a
generically nonsingular matrix, whose determinant is a multiple of
the resultant, so that the determinant degree with respect to the coefficients
of one polynomial equals that of the resultant.
For two univariate polynomials there are matrix formulae
named after Sylvester and B\'ezout, whose determinant equals the resultant;
we refer to them as detererminantal formulae.
Unfortunately, such determinantal formulae do not generally exist for more
variables, except for specific cases; this is the topic of our paper.

We consider the sparse (or toric) resultant, which exploits {\em a priori}
knowledge on the support of the equations.
Matrix formulae have been studied for systems where the
variables can be partitioned into groups so that every polynomial is
homogeneous in each group, i.e.\ {\em mixed multihomogeneous}, or
multigraded, systems. This study is an intermediate stage from the
theory of homogeneous and unmixed multihomogeneous systems, towards
fully exploiting arbitrary sparse structure.
Multihomogeneous systems are encountered in several areas, e.g.\
\cite{ChGoZh98,Mckelvey_totallymixed,EGL04}.
Few foundational works exist, such as
\cite{SCD07}, where bigraded systems are analyzed.
Our work continues that of
\cite{DE03,SZ94,WeZe}, where the unmixed case has been treated, and
generalizes their results to systems whose Newton polytopes
are scaled copies of one polytope. These are known as generalized unmixed
systems, and allow us to take a step towards systems
with arbitrary supports.  This is the first work that treats {\em
  mixed} multihomogeneous equations, and provides explicit resultant
matrices.

Sparse resultant matrices are of different types.  On the one end of
the spectrum are the {\em pure Sylvester-type} matrices, filled in by
polynomial coefficients; such are Sylvester's and Macaulay's matrices.
On the other end are the {\em pure B\'ezout-type} matrices, filled in
by coefficients of the {\em Bezoutian} polynomial. Hybrid matrices
contain blocks of both pure types.

We examine Weyman complexes (defined below), which generalize the
Cayley-Koszul complex and yield the multihomogeneous resultant as the
determinant of a complex.  These complexes are parameterized by a
\emph{degree vector} $\bm m $.  When the complex has two terms, its
determinant is that of a matrix expressing the map between these
terms, and equals the resultant.  In this case, there is a {\it
  determinantal\/} formula, and the corresponding vector $\bm m$ is
{\it determinantal\/}. The resultant matrix is then said to be
\emph{exact}, or {\em optimal}, in the sense that there is no
extraneous factor in the determinant.  As is typical in all such
approaches, including this paper, the polynomial coefficients are
assumed to be sufficiently generic for the resultant, as well as any
extraneous factor, to be nonzero.

In~\cite{WeZe}, the unmixed multihomogeneous systems for which a
determinantal formula exists were classified, but no formula was given;
see also~\cite[Sect.13.2]{GKZ}. Identifying explicitly the corresponding
morphisms and the vectors $\bm m$ was the focus of~\cite{DE03}.  The
main result of~\cite{SZ94} was to establish
that a determinantal formula of Sylvester type exists (for unmixed
systems) precisely when the condition of~\cite{WeZe} holds on the cardinalities
of the groups of variables and their degrees.  In \cite[Thm.2]{SZ94}
all such formulae are characterized by showing a bijection with the
permutations of the variable groups and by defining the corresponding
vector $\bm m$.  This includes all known Sylvester-type formulae, in
particular, of linear systems, systems of two univariate polynomials,
and bihomogeneous systems of 3 polynomials whose resultant is,
respectively, the coefficient determinant, the Sylvester resultant and
the classic Dixon formula.

In \cite{SZ94}, they characterized all determinantal Cay\-ley-Kos\-zul
complexes, which are instances of \emph{Weyman complexes} when all the
higher cohomologies vanish.  In \cite{DE03}, this characterization is
extended to the whole class of unmixed Weyman complexes.  It is also
shown that there exists a determinantal pure B\'ezout-type resultant
formula if and only if there exists such a Sylvester-type formula.  Explicit
choices of determinantal vectors are given for any matrix type, as
well as a choice yielding pure B\'ezout type formulae, if one exists.
The same work provides tight bounds for the coordinates of all
possible determinantal vectors and, furthermore, constructs a family of
(rectangular) pure Sylvester-type formulae among which lies the
smallest such formula.  This paper shall extend these results to
unmixed systems with scaled supports.

Studies exist, e.g.~\cite{ChGoZh98}, for computing hybrid formulae for
the resultant in specific cases.  In~\cite{AwaChkGoz05}, the Koszul
and Cech cohomologies are studied in the mixed multihomogeneous case
so as to define the resultant in an analogous way to the one used in
Section~\ref{resultants}.  In \cite{DADi01}, hybrid resultant formulae
were proposed in the mixed homogeneous case; this work is generalized
here to multihomogeneous systems.  Similar approaches are applied to
Tate complexes~\cite{DM08} to handle mixed systems.

The main contributions of this paper are as follows: Firstly, we
establish the analog of the bounds given in \cite[Sect.3]{DE03}; in
so doing, we simplify their proof in the unmixed case.  We
characterize the scaled systems that admit a determinantal formula,
either pure or hybrid.  If pure determinantal formulae exist, we
explicitly provide the $\bm m$-vectors that correspond to them.  In
the search for determinantal formulae we discover box domains that
consist of determinantal vectors thus improving the wide search for
these vectors adopted in~\cite{DE03}. We conjecture that a formula of
minimum dimension can be recovered from the centers of such boxes,
analogous to the homogeneous case.  

Second, we make the differentials in the Weyman complex explicit and
provide details of the computation.  Note that the actual construction
of the matrix, given the terms of the complex, is nontrivial.  Our
study has been motivated by \cite{DE03}, where similar ideas are used
in the (unmixed) examples of their Section~7, with some constructions
which we specify in Example~\ref{unmixedExamSylv}.  Finally, we
deliver a complete, publicly available \Maple\ package for the
computation of multihomogeneous resultant matrices.  Based on the
software of \cite{DE03}, it has been enhanced with new functions,
including some even for the unmixed case, such as the construction of
resultant matrices and the efficient storage of complexes.

The rest of the paper is organized as follows.  We start with sparse
multihomogeneous resultants and Weyman complexes in
Section~\ref{resultants} below.  Section~\ref{determinantal} presents
bounds on the coordinates of all determinantal vectors and classifies
the systems that admit hybrid and pure determinantal formulae;
explicit vectors are provided for pure formulae and minimum dimension
choices are conjectured.  In Section~\ref{MatConstr} we construct the
actual matrices; we present Sylvester- and B\'ezout-type constructions
that also lead to hybrid matrices. We conclude with the presentation
of our \Maple\ implementation along with examples of its usage.

Some of these results have appeared in preliminary form in~\cite{EmiMan09}.

\section{Resultants via complexes} \label{resultants}

We define the resultant, and connect it to complexes by homological
constructions.  Take the product $X:=\PP^{l_1} \times \cdots \times
\PP^{l_r}$ of projective spaces over an algebraically closed field
$\mathbb F$ of characteristic zero, for $r\in\NN$.  Its dimension
equals the number of affine variables $n=\sum_{k=1}^r l_k$.  We
consider polynomials over $X$ of scaled degree: their multidegree is a
multiple of a base degree $\bm d=(d_1,\ldots,d_r)\in\NN^r$, say $\deg
f_i = s_i\bm d$.  We assume $s_0\le \cdots \le s_n$ and
$\gcd(s_0,\dots,s_n)=1$, so that the data $\bm l,\bm d,\bm
s=(s_0,\dots,s_n)\in\NN^{n+1}$ fully characterize the system.  We
denote by $S(\bm d)$ the vector space of multihomogeneous forms of
degree $\bm d$ defined over $X$. These are homogeneous of degree $d_k$ in
the variables $\bm x_k$ for $k=1,\ldots,r$.
By a slight abuse of notation, we also write $S(d_k) \subset
\PP^{l_k}$ for the subspace of homogeneous polynomials in $l_k$
variables, of degree $d_k$.
A system of type $(\bm l,\bm d,\bm s)$ belongs to $V= S(s_0\bm
d)\oplus\cdots\oplus S(s_n\bm d)$.

\begin{Definition}
  Consider a generic scaled multihomogeneous system $\bm
  f=(f_0,\dots,f_n)$ defined by the cardinalities $\bm l\in\NN^r$,
  base degree $\bm d\in\NN^r$ and $\bm s\in\NN^{n+1}$.  The
  \emph{multihomogeneous} resultant $\mathcal R(f_0,\ldots,f_n)
  =\mathcal R_{\bm l,\bm d,\bm s}(f_0,\ldots,f_n)$ is the unique up to
  sign, irreducible polynomial of $\ZZ[V]$, which vanishes if and only if there
  exists a common root of $f_0,\dots,f_n$ in $X$.
\end{Definition}
This polynomial exists for any data $\bm l,\,\bm d,\,\bm s$, since it is an
instance of the sparse resultant.  It is itself multihomogeneous in
the coefficients of each $f_i$, with degree given by the
multihomogeneous B\'ezout bound:

\begin{lemma}\label{mBezout} The resultant polynomial is homogeneous
in the coefficients of each $f_i,\ i= 0,\dots,n$, with degree
\begin{equation*} \deg_{f_i}\mathcal R = \binom{n}{l_1,\dots,l_r}
\frac{ d_1^{l_1}\cdots d_r^{l_r}s_0\cdots s_n}{s_i}.
\end{equation*}
\end{lemma} 
\begin{proof}
The degree $\deg_{f_i}\mathcal R$ of $\mathcal R(\bm f)$ 
with respect to $f_i$ is the coefficient of $y_1^{l_1} \cdots y_r^{l_r}$
  in the new polynomial:
$$
\prod_{j\neq i}  (  s_jd_{1}y_1+ \cdots + s_jd_{r}y_r) 
=
\prod_{j\neq i}  s_j (  d_{1}y_1+ \cdots + d_{r}y_r) 
=
\frac{s_0s_1\cdots s_n}{s_i} (  d_{1}y_1+ \cdots + d_{r}y_r)^n. \label{mbezoutpoly}
$$
In~\cite[Sect.4]{SZ94} the coefficient of $y_1^{l_1} \cdots y_r^{l_r}$ 
in $(  d_{1}y_1+ \cdots + d_{r}y_r)^n$ is shown to be equal to
$$
\binom{n}{l_1,\dots,l_r} d_1^{l_1}\cdots d_r^{l_r}  ,
$$
thus proving the formula in the unmixed case. Hence the coefficient of
$y_1^{l_1} \cdots y_r^{l_r}$ 
in our case is this number multiplied by $\ds\frac{s_0s_1\cdots s_n}{s_i}$.
\end{proof}
This yields the total degree of the resultant, that is,
$\sum_{i=0}^n \deg_{f_i}\mathcal R$.

The rest of the section gives details on the underlying theory.  The
vanishing of the multihomogeneous resultant can be expressed as the
failure of a complex of sheaves to be exact.  This allows to construct
a class of complexes of finite-dimensional vector spaces whose
determinant is the resultant polynomial.  This definition of the
resultant was introduced by Cayley \cite[App.~A]{GKZ},
\cite{Wey94}.

For $\bm u\in\ZZ^r$, $H^q\left (X,\mathcal O_X(\bm u) \right)$ denotes
the $q$-th cohomology of $X$ with coefficients in the sheaf ${\mathcal
  O_X}(\bm u)$.  Throughout this paper we write for simplicity $H^q(\bm
u)$, even though we also keep the reference to the space whenever it
is different than $X$, for example $H^0(\PP^{l_k},u_k)$.
To a polynomial system $\bm f=(f_0,\dots,f_n)$ over $V$, we associate
a finite complex of sheaves $K_\bullet$ on $X$~:
\begin{equation}
  0 \to K_{n+1}\to \cdots \xrightarrow{\delta_2} K_{1} \xrightarrow{\delta_1} K_0 \xrightarrow{\delta_0} \cdots \to K_{-n}\to 0
\end{equation}
This complex (whose terms are defined in Definition~\ref{WeyComplex} below)
is known to be exact if and only if $f_0,\ldots,f_n$ share no zeros in $X$; it is
hence generically exact. When passing from the complex of sheaves to a
complex of vector spaces there exists a degree of freedom, expressed
by a vector $\bm m=(m_1,\ldots,m_r)\in\ZZ^r$.  For every given $\bm f$
we specialize the differentials $\delta_i : K_i\to K_{i-1}$,
$i=1-n,\dots,n+1$ by evaluating at $\bm f$ to get a complex of
finite-dimensional vector spaces. The main property is that the
complex is exact if and only if $\mathcal R(f_0,\ldots,f_n) \not= 0$
\cite[Prop.1.2]{Wey94}.

The main construction that we study is this complex, which we define
in our setting.  It extends the unmixed case, where for given $p$ the
direct sum collapses to $\binom{n+1}{p}$ copies of a single cohomology
group.
\begin{Definition}\label{WeyComplex}
For  $\bm m\in\ZZ^r$, $\nu=-n,\dots,n+1$ and $p=0,\dots,n+1$ set
\begin{align*}
K_{\nu,p}
= & \bigoplus_{ 0\leq i_1<\cdots<i_p\leq n} H^{p-\nu}\left(\bm m -\sum_{\theta=1}^p s_{i_\theta} \bm d  \right)
\end{align*}
where the direct sum is over all possible indices $i_1< \cdots<
i_p$. The \emph{Weyman complex} $K_\bullet=K_\bullet(\bm l,\bm d,\bm
s,\bm m)$ is generically exact and has terms $\ds K_{\nu} =
\bigoplus_{p=0}^{n+1} K_{\nu,p}$.
\end{Definition}

This generalizes the classic \emph{Cayley-Koszul} complex. The
determinant of the complex can be expressed as a quotient of products
of minors from the $\delta_i$. It is invariant under different
choices of $\bm m\in\ZZ^r$ and equals the multihomogeneous resultant
$\mathcal R(f_0,\dots,f_n)$.

\subsection{Combinatorics of K$_\bullet$} \label{complex}

We present a combinatorial description of the terms in our complex,
applicable to the unmixed case as well. For details on the
co-homological tools that we use, see \cite{GKZ}.

By the K\"unneth formula, we have the decomposition
\begin{equation} \label{Kunneth}
H^{q}\left(\bm\alpha \right) =  \bigoplus_{j_1+\cdots+j_r=q}^{j_k\in\{0,l_k\}}
\bigotimes_{k=1}^r H^{j_k}\left(\PP^{l_k}, \alpha_k \right),
\end{equation}
where $q=p-\nu$ and the direct sum runs over all integer sums
$j_1+\cdots+j_r=q, j_k\in\{0,l_k\}$.  In particular, $H^0(\PP^{l_k},
\alpha_k)$ is isomorphic to $S(\alpha_k)$, the graded piece of
$\PP^{l_k}$ in degree $\alpha_k$ or, equivalently, the space of all
homogeneous polynomials in $l_k+1$ variables with total degree
$\alpha_k$, where $\bm \alpha= \bm m-z\bm d \in\ZZ^r$ for $z\in\ZZ$.

By Serre duality, for any $\bm \alpha \in \ZZ^r$, we know that
\begin{equation} \label{Serre}
H^q ( \bm \alpha)  \simeq H^{n-q}( -\bm l-\boldsymbol 1 - \bm \alpha)^*,
\end{equation}
where ${}^*$ denotes dual, and ${\bf 1}\in\NN^r$ a vector full of ones.
Therefore $H^j(\alpha_k)^* \simeq H^{l_k-j}(-\alpha_k-1-l_k)$.

Furthermore, we identify $H^{l_k}(\PP^{l_k}, \alpha_k)$ as the dual space
$ S(-\alpha_k - l_k -1 )^*$.
This is the space of linear functions $\Lambda:S(\alpha_k)\to \mathbb
F$.  Sometimes we use the ne\-ga\-ti\-ve symmetric powers to interpret
dual spaces, see also \cite[p.576]{WeZe}.  This notion of duality is
naturally extended to the direct sum of cohomologies: the dual of a
direct sum is the direct sum of the duals of the summands.  The next
proposition (Bott's formula) implies that this dual space is
nontrivial if and only if $-\alpha_k - l_k -1\ge 0$.
\begin{proposition}{\rm \cite{Bot57}}\label{Botts}
  For any $\bm \alpha\in\ZZ^r$ and $k\in\{1,\dots,r\}$,\\
  (a) $H^{j}(\PP^{l_k}, \alpha_k ) =0 ,\ \forall j \not= 0, l_k$,\\
  (b) $H^{l_k}(\PP^{l_k}, \alpha_k)\ne 0\Leftrightarrow \alpha_k< -l_k$, $\dim H^{l_k}(\PP^{l_k}, \alpha_k) = {\binom{-\alpha_k -1}{l_k}}$.\\
  (c) $H^{0}(\PP^{l_k}, \alpha_k)\ne 0\Leftrightarrow \alpha_k \ge 0$,
  $\dim H^{0}(\PP^{l_k}, \alpha_k) = {\binom{\alpha_k +l_k}{l_k}}$.
\end{proposition}

\begin{Definition}\label{Dcrit}
Given $\bm l,\bm d\in\NN^r$ and $\bm s\in\NN^{n+1}$,
define the {\em critical degree} vector
$\rho\in\NN^r$ by $\bm \rho_k:= d_k \sum_{\theta=0}^n s_{\theta} -l_k -1,$ for all $k=1,\ldots r$.
\end{Definition}

The K\"unneth formula~(\ref{Kunneth}) states that $H^q(\bm\alpha)$ is
a sum of products.  We can give a better description:
\begin{lemma}\label{qth_Coh}
  If $H^{q}(\bm \alpha)$ is nonzero, then it is equal to a product
  $H^{j_1}(\PP^{l_k},\alpha_k)\otimes\cdots\otimes
  H^{j_r}(\PP^{l_k},\alpha_k)$ for some integers $j_1,\dots,j_r$ with
  $j_k\in\{0,l_k\},\ \sum_{k=1}^r j_k=q$.
\end{lemma}
\begin{proof}
  By Proposition~\ref{Botts}(a), only $H^0(\PP^{l_k},\alpha_k)$ or
  $H^{l_k}(\PP^{l_k},\alpha_k) $ may be nonzero.  By
  Proposition~\ref{Botts}(b,c) at most one of them appears.
\end{proof}

Combining Lemma~\ref{qth_Coh} with Definition~\ref{WeyComplex} and~(\ref{Kunneth}) we get
\begin{equation}
  K_{\nu,p} = \bigoplus_{ 0\leq i_1<\cdots<i_p\leq n}
  \bigotimes_{k=1}^r H^{j_k}\left(\PP^{l_k}, m_k- \sum_{\theta=1}^p s_{i_\theta} d_k \right)
\end{equation}
for some integer sums $j_1+\cdots+j_r=p-\nu,\ j_k\in\{0,l_k\}$ such
that all the terms in the product do not vanish.  Consequently, $\ds
\dim H^q\left(\bm \alpha \right)
= 
\prod_{k=1}^r \dim H^{j_k}\left(\PP^{l_k}, \alpha_k\right).  $ The
dimension of $K_{\nu,p}$ follows by taking the sum over all
$\bm \alpha=\bm m-\sum_{\theta=1}^p s_{i_\theta} \bm d$, for all combinations
$\{i_1<\cdots< i_p\}\subseteq\{0,\dots,n\}.$

Throughout this paper we denote $[u,v]:=\{u,u+1,\dots,v\}$; given
$p\in[0,n+1]$, the set of possible sums of $p$ coordinates out of
vector $\bm s$ is 
$$
 S_p:=\left\{ \sum_{\theta=1}^p s_{i_\theta}\ :\
  0\leq i_1<\cdots<i_p\leq n \right\}
$$ 
and by convention $S_0=\{0\}$. By Proposition~\ref{Botts}, the set of
integers $z$ such that both $H^0(\PP^{l_k},m_k-zd_k)$ and
$H^{l_k}(\PP^{l_k},m_k-zd_k)$ vanish is:
$$
P_k:=\left(\frac{m_k}{d_k},\frac{m_k+l_k}{d_k}  \right]\cap\ZZ.
$$
We adopt notation from \cite{WeZe}: for $u\in \ZZ$, $P_k<u \iff u >
\frac{m_k+l_k}{d_k} $ and $P_k>u \iff u \leq \frac{m_k}{d_k} $. Note
that we use this notation even if $P_k=\varnothing$.
As a result, the $z\in\ZZ$ that lead to a nonzero
$H^{j_k}(\PP^{l_k},m_k-zd_k)$, for $j_k=l_k$ or $j_k=0$, and
$p\in[0,n+1]$, lie in:
\begin{equation} \label{Qpsets} Q_p= S_p \setminus \cup_{1}^r P_k,
  \mbox{ and } Q=\cup_{p=0}^{n+1}Q_p.
\end{equation}
Now $\#P_k\le l_k$ implies $\#(\cup_k P_k)\le n$.  So $\#(\cup_p
S_p)\ge n+2$ implies $\#Q\ge 2$.
We define a function $q : Q\to [0,n]$ by
\begin{equation} \label{qmap} q(z):=\sum_{P_k<z}l_k.
\end{equation}
Observe that $H^j(X,\bm m-z\bm d)\neq 0 \iff z\in Q$ and $j=q(z)$; also the
system is unmixed if and only if $S_p=\{p\}$. Clearly $1\leq \#S_p \leq
\binom{n+1}{p}$, the former inequality being strict for $\bm s\neq
\boldsymbol 1\in\NN^{n+1}$ and $p\neq 0,n+1$.

The following lemma generalizes \cite[Prop.2.4]{WeZe}.

\begin{lemma} \label{KnpNz} Let $\nu\in\ZZ$, $p\in\{0,\dots,n+1\}$ and
  $K_{\nu,p}$ given by Definition~\ref{WeyComplex}; then $\ds K_{\nu,p}\neq
  0 \iff \nu\in\left\{p-q(z)\ :\ z\in Q_p\right\}$.
\end{lemma}
\begin{proof} 
  Assuming $K_{\nu,p}\neq 0$, there exists a nonzero summand
  $H^{p-\nu}(\bm m-z\bm d)\neq 0$. By Lemma~\ref{qth_Coh} it is equal to
  $H^{j_1}(\PP^{l_k},m_k-zd_k)\otimes\cdots\otimes
  H^{j_r}(\PP^{l_k},m_k-zd_k)\neq 0,\ j_k\in\{0,l_k\}$ and
  $$ p-\nu=\sum_{k=1}^r j_k= \sum_{P_k<z}l_k \Rightarrow \nu=p-\sum_{P_k<z}l_k . $$
  Conversely, if $\nu\in \{p-q(z) : z\in Q_p \}$ then $Q_p\neq
  \varnothing$. Now $z\in Q_p$ implies $z\notin P$, which means
  $H^{q(z)}(\bm m-z\bm d)\neq 0$, the latter being a summand of
  $K_{\nu,p}$.
\end{proof}

One instance of the complexity of the mixed case is that in the
unmixed case, given $p\in [0,n+1]$, there exists at most one integer
$\nu$ such that $K_{\nu,p}\neq 0$.

All formulae (including determinantal ones) come in dual pairs, thus
generalizing \cite[Prop.4.4]{DE03}.
\begin{lemma} \label{dualVecs} Assume $\bm m, \bm m'\in \ZZ^r$ satisfy
  $\bm m + \bm m'= \bm \rho$, where $\bm \rho$ is the critical degree
  vector.
  Then, $K_\nu(\bm m)$ is dual to $K_{1-\nu}(\bm m')$ for all $\nu \in \ZZ.$
  In particular, $\bm m$ is determinantal if and only if $\bm m'$ is
  determinantal, yielding matrices of the same size, namely
  $\dim(K_0(\bm m)) = \dim (K_1(\bm m')).$
\end{lemma}
\begin{proof}
  Based on the equality $\bm m +\bm m'= \bm \rho$ we deduce that for
  all $J\subseteq [0,n+1]$, it holds that $\bm m' -\sum_{i\in J}
  s_i\bm d = -\bm l - \boldsymbol 1 - (\bm m- \sum_{i\notin J}s_i\bm
  d).$ Therefore, for all $q=0,\dots,n$, Serre's duality (\ref{Serre})
  implies that $H^q(X, \bm m'-\sum_{i\in J}s_i\bm d)$ and $H^{n-q}(X,
  \bm m- \sum_{i\notin J}s_i\bm d )$ are dual.

  Let $\#J=p$ and $\nu=p-q$; since $(n+1-p) - (n-q) = 1 -(p-q)=1-\nu$,
  we deduce that $K_{\nu, p}(\bm m)$ is dual to $K_{1-\nu, n+1-p}(\bm m')
  $ for all $p\in[0,n+1]$ which leads to $K_\nu(\bm m)^*\simeq
  K_{1-\nu}(\bm m')$ for all $\nu \in \ZZ$, as desired.  In
  particular, $K_{-1}(\bm m) \simeq K_2^*(\bm m')$ and $K_0(\bm m)
  \simeq K_1^*(\bm m')$, the latter giving the matrix dimension in the
  case of determinantal formulae.
\end{proof}

\section{Determinantal formulae} \label{determinantal}

This section focuses on formulae that yield square matrices expressing
the resultant without extraneous factors and prescribes the
corresponding {determinantal} $\bm m$-vectors.

Determinantal formulae occur only if there is exactly one nonzero
differential, so the complex consists of two consecutive nonzero
terms. The determinant of the complex is the determinant of this
differential.  We now specify this differential; for the unmixed case
see \cite[Lem.3.3]{WeZe}.
\begin{lemma} \label{det01} If $\bm m\in\ZZ^r$ is determinantal then the
  nonzero part of the complex is $\delta_1: K_1\to K_0$.
\end{lemma}
\begin{proof}
  The condition that $\bm m$ is determinantal is equivalent to the
  fact that $N:= \{p-q(z):\ z\in Q_p, p\in[0,n+1] \}$ consists of two
  consecutive integers.

  Let $z_1=\min Q < z_2=\max Q$, since $\#Q\ge 2$.  There exist
  $p_1,p_2$ with $p_1< p_2$ such that
  $z_1=\sum_{\theta=1}^{p_1}s_{i_{\theta}}$ and
  $z_2=\sum_{\lambda=1}^{p_2}s_{j_{\lambda}}$ where the indices are
  sub-sequences of $[0,n]$, of length $p_1$ and $p_2$ resp.  The $p_1$
  integers
  $$
  0,s_0,s_0+s_1,\dots,s_0+\cdots+s_{p_1-2}\in\ZZ
  $$
  are distinct, smaller than $z_1$, hence belong to $\cup_{P_k<z_1}
  P_k$. Also, it is clear that, for all $k\in[1,r]$, $\#P_k\leq \lceil
  l_k/d_k \rceil \leq l_k$ thus
  \begin{equation}\label{p1qz1}
    p_1 \leq \# \bigcup_{P_k<z_1} P_k  \leq \sum_{P_k<z_1} \# P_k \leq q(z_1).
  \end{equation}
  This means $p_1-q(z_1) \leq 0$.
  Similarly, the $n+1-p_2$ integers $s_n+\cdots +s_0,\dots,
  s_n+\cdots+ s_{n-p_2}$ are distinct, larger than $z_2$, hence
   belong to $\cup_{P_k>z_2} P_k$, so:
   \begin{equation}\label{p2qz2}
     n-p_2+1 \leq \# \bigcup_{P_k>z_2} P_k \leq \sum_{P_k>z_2} \# P_k
     \leq \sum_{P_k>z_2}l_k.
   \end{equation}
   This means $n+1-p_2\leq n- q(z_2)$, thus $p_2 - q(z_2) \geq 1$.
  Hence there exists a positive integer in $N$; from~(\ref{p1qz1}) we
  must have a non-positive integer in $N$. Since $\#N=2$ and the
  integers of $N$ are consecutive we deduce that $N =\{0,1\}$.
\end{proof}

\begin{corollary} \label{detCor} If $\bm m\in \ZZ^r$ is determinantal,
  then equality holds in~(\ref{p1qz1}).  In particular, for $k
  \in[1,r]$ s.t.\ either $P_k<z_1$ or $P_k>z_2$, we have $ \# P_k =
  l_k $ and any two such $P_k$ are disjoint.
\end{corollary}
\begin{proof}
  Lemma~\ref{det01} combined with Lemma~\ref{KnpNz} imply
  $p_1-q(z_1)\geq 0$, and~(\ref{p1qz1}) implies $p_1-q(z_1) \leq 0$,
  hence we deduce $p_1-q(z_1)= 0$. Now equality in~(\ref{p1qz1}) gives
  $\sum_{P_k<z_1} \# P_k = q(z_1)= \sum_{P_k<z_1}l_k $; combining with
  $\#P_k\leq l_k$ we deduce $\#P_k=l_k$ for all $k$ in this
  sum. Similarly for $n+1-p_2= \sum_{P_k>z_2} l_k = n-q(z_2).$
\end{proof}

\subsection{Bounds for determinantal vectors}

We generalize the bounds in \cite[Sect.3]{DE03} to the
mixed case, for the coordinates of all determinantal $\bm m-$vectors.
We follow a simpler and more direct approach based on a
global view of determinantal complexes.

\begin{lemma} \label{PkContained} If a vector $\bm m\in\ZZ^r$ is
  determinantal then the corresponding $\bigcup_{1}^r P_k$ is
  contained in $\left[ 0, \sum_0^n s_i \right]$.
\end{lemma}
\begin{proof}
  It is enough to establish that $P_k>0$ and $P_k< \sum_0^n s_i$ for
  all $k\in[1,r]$. Proof by contradiction: Let $\bm m$ be a determinantal
  vector, and $P=\cup_0^k P_k$.  Let $z_1,z_2$ as in the proof of
  Lemma~\ref{det01}. If $z_1< P_k< z_2$ it is clear that $ 0\leq z_1 <
  P_k < z_2 \leq \sum_0^n s_i \Rightarrow P_k\subseteq \left[ 0,
    \sum_0^n s_i \right]$.

  If $z_2< P_k$, Corollary~\ref{detCor} implies $\# \bigcup_{P_k>z_2} P_k =
  \#R$, where $R:=\{ s_n+\cdots +s_0,\dots, s_n+\cdots+ s_{n-p_2}
  \}$. By the definition of $z_2$, $R \subseteq \bigcup_{P_k>z_2}P_k$,
  thus $\bigcup_{P_k>z_2}P_k = R\subseteq \left[ 0, \sum_0^n s_i
  \right]$.  Similarly, $\bigcup_{P_k<z_1}P_k = \{
  0,s_0,s_0+s_1,\dots,s_0+\cdots+s_{p_1-2} \}\subseteq \left[ 0,
    \sum_0^n s_i \right]$, which proves the lemma for $P_k<z_1$.
\end{proof}

The bound below is proved in \cite[Cor.3.9]{DE03} for the unmixed
case.  They also show with an example that this bound is tight 
with respect to individual coordinates.  We give an independent,
significantly simplified proof, which extends that result to the
scaled case.
\begin{theorem}\label{mBounds}
  For determinantal $\bm m\in\ZZ^r$, for all $k$ we have 
$$
  \max\{-d_k,-l_k\} \le m_k \le d_k \sum_0^n s_i -1 +
  \min\{d_k-l_k,0\} .
$$
\end{theorem}
\begin{proof}
  Observe that by Lemma~\ref{det01} there are no $k\in[1,r]$ such that
  $P_k<0$ or $P_k> \sum_0^n s_i$.  Combining this fact with
  Lemma~\ref{PkContained}, we get
  \begin{equation}
  m_k/d_k\geq -1 \ \text{ and }\ (m_k+l_k)/d_k< 1+\sum_0^n s_i \label{Ineq1}
  \end{equation}
  for all $k\in[1,r]$. Furthermore, the sets $P_k,\, k\in[1,r]$ can be
  partitioned into two (not necessarily non-empty) classes, by
  considering the integers $z_1,\, z_2$ of Lemma~\ref{det01}:
  \begin{itemize}
  \item $P_k<z_1$ or $P_k>z_2$, with cardinalities $\#P_k=l_k$.
  \item $ z_1< P_k < z_2 $, without cardinality restrictions (possibly
    empty).
  \end{itemize}
  Taking into account that
  $P_k=\left(\frac{m_k}{d_k},\frac{m_k+l_k}{d_k} \right]\cap\ZZ$ we
  get
  \begin{equation}
    (m_k+l_k)/d_k\geq 0\ \text{ and }\ m_k/d_k<\sum_0^n s_i \label{Ineq2}
  \end{equation}
  for all
  $k\in[1,r]$.
\end{proof}

Our implementation in Section~\ref{implement} conducts a search in the
box defined by the above bounds. For each $\bm m$ in the box, the
dimension of $K_2$ and $K_{-1}$ is calculated; if both are zero the
vector is determinantal. Finding these dimensions is time consuming;
the following lemma provides a cheap necessary condition to check
before calculating them.

\begin{lemma} \label{nessCond} If $\bm m\in\ZZ^r$ is determinantal then there
  exist indices $k,k'\in[1,r]$ such that $m_k<d_k(s_{n-1}+s_n)$ and
  $m_{k'}\geq d_{k'} \sum_0^{n-2} s_i - l_{k'}$.
\end{lemma}
\begin{proof} 
  If for all $k$, $m_k/d_k\geq s_0+s_1$ then $q(s_{n-1}+s_n)=0$
  by~(\ref{Qpsets}), so for $p=2$ we have $p-q(s_{n-1}+s_n)=2-0=2$
  which contradicts the fact that $m$ is determinantal. Similarly, if
  for all $k$, $(m_k+l_k)/d_k < \sum_0^{n-2} s_i \Rightarrow
  q\left(\sum_0^{n-2} s_i\right)=n$ and for $p=n-1$ we have $p -
  q\left(\sum_2^{n} s_i\right)=(n-1)-n=-1$, which is again infeasible.
\end{proof}

\subsection{Characterization and explicit vectors}

A formula is determinantal if and only if $K_2=K_{-1}=0$.
In this section we provide necessary and sufficient conditions for the data $\bm l,\bm d,\bm s$ to
admit a determinantal formula; we call this data
\emph{determinantal}. Also, we derive multidimensional integer
intervals (boxes) that yield determinantal formulae and conjecture
that minimum dimension formulae appear near the center of these
intervals. 

\begin{lemma}\label{disjointformula}
  If $\bm m\in\ZZ^r$ is a determinantal vector for the data $\bm l,\bm
  d,\bm s$, then this data admits a determinantal vector $\bm m'\in\ZZ^r$
  with $P_k\cap P_{k'}=\varnothing$ for all $k,k'\in[1,r]$.
\end{lemma}
\begin{proof} 
  Suppose $m_i/d_i\le m_j/d_j$.
  Let $P_i(\bm m)\cap P_{j}(\bm m)=[u,v] \subset\ZZ$. Set $m'_j= m_j + t d_j$
  where $t\in\ZZ$ is the minimum
  shift so that $P_i(\bm m')\cap P_{j}(\bm m')=\varnothing$ and $P_j(\bm m')$
  satisfies Theorem~\ref{mBounds}.  For all $k\ne j$, let $m'_k=m_k$.

Any vector in $\ZZ^r$ defines a nontrivial complex, since $Q\ne \varnothing$.
In particular, $\bm m'$ is determinantal because $P(\bm m)\subseteq P(\bm m')$, i.e.\ no new terms are introduced,
but possibly some terms vanish.
Repeat until all $P_k\cap P_{k'}=\varnothing$.
\end{proof}

Let $\sigma:[1,r]\to [1,r]$ be any permutation. One can identify at
most $r!$ classes of determinantal complexes, indexed by the
permutations of $\{1,\dots,r\}$. This classification arises if we look
at the nonzero terms that can occur in the complex, provided that the
sets $P_k$ satisfy
$$
P_{\sigma(1)} \leq P_{\sigma(2)} \leq \cdots \leq P_{\sigma(r)}
$$
where we set $P_i \leq P_j \ \iff \ m_i/d_i \leq m_j/d_j $.
Any given $\bm m$ defines these sets, as well as an ordering between them.
This fact allows us to classify determinantal $\bm m$-vectors and the
underlying complexes.

For this configuration, expressed by $\sigma$, the only nonzero
summands of $K_\nu$ can be $K_{\nu,\nu+q}$ where $q$ takes values in
the set $\{0,l_{\sigma(1)}, l_{\sigma(1)}+l_{\sigma(2)},\dots, n \}$.
To see this, observe that $q(z)=\sum_{P_k<z} l_k$, $z\in \NN$ cannot
attain more than $r+1$ distinct values; so if the relative ordering of
the $P_k$ is fixed as above, then these are the only possible values
of $q$.  This leads us to the following description of $K_2$ and
$K_{-1}$:
\begin{equation}\label{Ksigma}
K_2^{\sigma}= \bigoplus_{k=1}^r K_{2,2+\sum_{i=1}^{k-1} l_{\sigma(i)}}
\ \ , \ \ K_{-1}^{\sigma}= \bigoplus_{k=1}^r K_{-1,-1+\sum_{i=1}^{k}
  l_{\sigma(i)}}
\end{equation}
As a side remark, note that the proof of Lemma~\ref{dualVecs} implies
that the dual of $K_\nu^{\sigma}(\bm m)$ is $K_{1-\nu}^{\tau}(\bm
\rho-\bm m)$ where $\tau$ is the permutation s.t.\ $\tau(i):= r+1 -
\sigma(i)$.

Let $ \pi[k]:= \sum_{\pi(i)\leq \pi(k)} l_i $. If $\pi=\text{Id}$ this
is $\text{Id}[k]=l_1+\cdots+l_k$.  We now characterize determinantal
data:
\begin{theorem}\label{Hdeter}
  The data $\bm l,\bm d,\bm s$ admit a determinantal formula if and only if
  there exists $\pi:[1,r]\to [1,r]$ s.t.
  $$ 
  d_k \sum_{n- \pi[k]+2}^n s_i - l_k < d_k \sum_0^{\pi[k-1]+1} s_i ,\
  \forall k .
  $$
\end{theorem}
\begin{proof}
  We assume without loss of generality that $\pi=\id$.  This is not
  restrictive, since if $\pi\ne\id$ then we can re-number the variable
  groups such that $k':=\pi^{-1}(k)$. Hence if we set
  $L_k:=\sum_0^{\id[k-1]+1}s_i= \min S_{\id[k-1]+2}$ and $R_k:=
  \sum_{n- \id[k]+2}^n s_i=\max S_{\id[k]-1}$ then the relations become:
  $$
  d_k R_k - l_k < d_k L_k ,\ \forall k .
  $$
  Throughout this proof, whenever we use non-positive indices $j\le 0$
  for $l_{j}$ or $\id[j]$, these quantities will be zero, and the
  results in this case are straightforward to verify.  Also, note that
  the dual complex is given by the ``reversed'' permutation, and in
  particular, $K_{-1}^* \simeq K_2$, therefore any results on the
  nullity of $K_{-1}$ can be directly used to prove the nullity of $K_{2}$.

  ($\Leftarrow$) Assume that the inequalities hold. Then for all $k$
  there exists an integer $m_k$ such that
  \begin{equation}\label{mbox}
  d_k R_k - l_k \le  m_k  \le d_k L_k - 1 . 
  \end{equation}
  Let $\bm m=(m_1,\dots,m_r)$. We shall prove that this vector gives a
  determinantal formula; it suffices to show that for all $k\in[1,r]$,
  $K_{2,2+\id[k-1]}=K_{-1,-1+\id[k]}=0$, since in~(\ref{Ksigma}) we have
  $\sigma=\pi^{-1}=\id$.
  
  \textbullet\ If $l_k\ge 3$, we have 
  \begin{equation}\label{cardinality}
    \id[k]-1=\id[k-1]+l_k-1\geq \id[k-1]+2 . 
  \end{equation}

  Thus $L_k \leq R_k$, since $s_i\ge 1$; also our
  hypothesis~(\ref{mbox}) translates into the inclusion
  \begin{equation}\label{Hcase1}
    \big[ L_k  , R_k \big] \subseteq \left(\frac{m_k}{d_k}, 
    \frac{m_k+l_k}{d_k} \right] = P_k   .
  \end{equation}
  Now, by~(\ref{cardinality}) we derive 
  $$
  \min S_{\id[k-1]+2} = L_k 
  \le  \min S_{\id[k]-1}
  \ \ \text{and}\ \ \
  \max S_{\id[k]-1} = R_k
  \ge \max S_{\id[k-1]+2}
  $$
  so~(\ref{Hcase1}) implies $ S_{\id[k-1]+2} \subseteq P_k$ as well as
  $S_{\id[k]-1}\subseteq P_k$ and thus $K_{2,2+ \id[k-1]}=0$, $
  K_{-1,-1+\id[k]}=0$ by Proposition~\ref{Botts}.

  \textbullet\ If $l_k \le 2$, it is 
   $\id[k]-1=\id[k-1]+l_k-1\leq \id[k-1]+1$.
  In this case we will prove $K_{2,2+ \id[k-1]}=0$, $
  K_{-1,-1+\id[k]}=0$ using Lemma~\ref{KnpNz}.

  Let $z\in Q_p$ for $p={-1+\id[k]}$. From $R_k\leq(m_k+l_k)/d_k$ it
  is clear that $P_k \not< z$, thus $q(z)\leq \id[k-1]$. Also,
  $$
  \id[k-2]+2= \id[k]-l_{k-1} -l_k + 2
  \leq \id[k]-1 . 
  $$
  where the last inequality is taken under the assumption
  $\max\{l_k,l_{k-1}\}\ge 2$. We treat the case $l_k=l_{k-1}=1$ seperately. Hence
  $$
  z \ge
 \min S_{-1+\id[k]} \ge \min
  S_{2+\id[k-2]}= L_{k-1} > m_{k-1}/d_{k-1}      . 
  $$  
  This implies $P_{k-1}<z \Rightarrow q(z)\geq \id[k-1]$. We
  conclude that $q(z)= \id[k-1]$ and thus
  \begin{equation}\label{Hcase2a}
    p-q(z) =(-1+\id[k])-\id[k-1]=-1+l_{k} \in [0,1] .
  \end{equation}
  By Lemma~\ref{KnpNz} we see that $K_{-1,-1+\id[k]}=0$, since
  $p-q(z)\ne -1$. 

  To complete the proof, suppose $l_k=l_{k-1}=1$ and
  $p=\id[k]-1$. We get $\id[k]-1=\id[k-2]+1$
  and therefore $z > P_{k-2} \Rightarrow q(z)\geq
  \id[k-2]$. Recall that $q(z)$ is also upper bounded by $\id[k-1]$.
  We derive that for $z\in Q_{p}$ it holds $p-1\le q(z)\le p$, therefore $p - q(z)\in
  [0,1]$, and again by Lemma~\ref{KnpNz}, $K_{-1,-1+\id[k]}=0$.

  As already pointed out, by using duality one can see that, for $z'\in
  Q_{2+\id[k-1]}$, it holds $p-q(z')\ne 2$, therefore $K_{2,2+\id[k-1]}=0$.

  ($\Rightarrow$) Suppose that $\bm m\in\ZZ^r$ is determinantal,
  namely $K_2(\bm m)=K_{-1}(\bm m)=0$.  
  Lemma~\ref{disjointformula} implies that we may assume the
  sets $P_j$ are pairwise disjoint.  By a permutation of the
  variable groups we also assume that the $P_j$ sets induced by $\bm m$
  satisfy
  $$
  P_{1} \leq P_{2} \leq \cdots \leq P_{r} .
  $$
  The sets $P_j$ have to be distributed along $I:=[0, \sum_0^n s_i ]$
  (Lemma \ref{PkContained}) and the connected components of
  $I\setminus \cup P_j$ are subsets of $S_p\cup S_{p+1}$,
  $p\in[0,n+1]$ since they define a determinantal complex. In
  particular, for $p=\id[k]-1$ we get $P_{k-1} \le S_p\cup S_{p+1} \le
  P_{k+1}$. Now the definition of $R_k$ as an element of
  $S_{\id[k]-1}$ implies $P_{k-1}<R_k<P_{k+1}$, i.e.  we have the
  implications (similarly for $L_k$):
  \begin{align}\label{impl1}
  R_k\in \cup_1^rP_j \ \Longrightarrow \ R_k\in P_k  
  \ \ \text{ and }\ \ 
  L_k\in \cup_1^rP_j \ \Longrightarrow \ L_k\in P_k  .
  \end{align}
 
  Suppose $m_k<d_kR_k-l_k$, or equivalently
  $\ds\frac{m_k+l_k}{d_k}<R_k$. Then $R_k\notin P_k$, hence by
  (\ref{impl1}) we must have $R_k\notin \cup_j P_j$ which leads to $z=R_k\in
  Q_p$, $p={\id[k]-1}$. This implies
  $$
  q(z)\geq\id[k]\ \ \Rightarrow \ \ p-q(z) \leq \id[k]-1-\id[k]=-1 \ \Rightarrow \ K_{-1}\ne 0,
  $$
  which is a contradiction.
  In the same spirit, if $m_k\geq d_kL_k$, we are led to $z'=L_k\in Q_p$, $p={\id[k-1]+2}$, then
  $$
  q(z')\leq\id[k-1]\ \ \Rightarrow \ \ p-q(z') \geq p -\id[k-1]=2 \ \Rightarrow \ K_{2}\ne 0,
  $$
  which again contradicts our hypothesis on $\bm m$.

  We conclude that any coordinate $m_k$ of $\bm m$ must satisfy $d_k
  R_k - l_k \le m_k < d_k L_k$, hence the existence of $\bm m$
  implies the inequality relations we had to prove.
\end{proof}

\begin{corollary}\label{DetBoxes}
  For any permutation $\pi:[1,r]\to [1,r]$, the vectors $\bm m\in\ZZ^r$
  contained in the box
  $$ 
  d_k \sum_{n- \pi[k]+2}^n s_i - l_k \leq m_k \leq d_k
  \sum_0^{\pi[k-1]+1} s_i - 1
  $$
  for $k=1,\dots, r$ are determinantal.
\end{corollary}

It would be good to have a characterization that does not depend on
the permutations of $[1,r]$; this would further reduce the time needed
to check if some given data is determinantal. One can see that if
$r\leq 2$ an equivalent condition is $\ds d_k \sum_{n-l_k+2}^n s_i -
l_k < d_k ( s_0 + s_1 )$ for all
$k\in[1,r]$; see~\cite[Lem.5.3]{DADi01} for the case $r=1$.  It turns
out that for any $r\in\NN$ this condition is \emph{necessary} for the
existence of determinantal vectors, but not always sufficient: the
smallest counterexample is $\bm l=(1,2,2),\, \bm d=(1,1,1),\,
\bm s=(1,1,1,1,2,3)$: this data is not determinantal, although the
condition holds. In our implementation this condition is used as a
filter when checking if some data is determinantal. Also,
\cite[Cor.5.5]{DADi01} applies coordinate-wise: if for some $k$,
$l_k\geq 7$ then a determinantal formula cannot possibly exist unless
$d_k=1$ and all the $s_i$'s equal $1$, or at most, $s_{n-1}=s_{n}=2$,
or all of them equal $1$ except $s_n=3$.

We deduce that there exist at most $r!$ boxes, defined by the above
inequalities that consist of determinantal vectors, or at most $r!/2$
matrices up to transpose. One can find examples of data with any even
number of nonempty boxes, but by Theorem~\ref{Hdeter} there exists at
least one that is nonempty.

If $r=1$ then a minimum dimension formula lies in the center of an
interval~\cite{DADi01}.  We conjecture that a similar explicit choice
also exists for $r>1$.  Experimental results indicate that minimum
dimension formulae tend to appear near the \emph{center} of the
nonempty boxes:
\begin{conjecture} \label{concenters} If the data $\bm l,\bm d,\bm s$
  is determinantal then determinantal degree vectors of minimum matrix
  dimension lie \emph{close} to the center of the nonempty boxes of
  Corollary~\ref{DetBoxes}.
\end{conjecture}

We conclude this section by treating the homogeneous case, as an example.

\begin{example} \label{dandrea} The case $r=1$, arbitrary degree, has
  been studied in \cite{DADi01}. We shall formulate the problem in our
  setting and provide independent proofs.  Let $n,d\in\ZZ$,
  $\bm s\in\ZZ^{n+1}_{>0}$. This data define a scaled homogeneous system
  in $\PP^n$; given $m\in\ZZ$, we obtain $P=\left( \frac m d,
    \frac{m+n}{d}\right]\cap\ZZ$. In this case there exist only zero
  and $n$th cohomologies; zero cohomologies can exist only for
  $\nu\geq 0$ and $n$th cohomologies can exist only for $\nu\leq
  1$. Thus in principle both of them exist for $\nu\in\{0,1\}$. Hence,
  $$
    K_\nu=\left\{\begin{array}{ll}
    K_{\nu,\nu}, &\ \ \ 1<\nu\leq n+1 \\
    K_{\nu,\nu}\oplus K_{\nu,n+\nu}, &\ \ \ 0\leq \nu \leq 1\\
    K_{\nu,n+\nu}, &\ \ \ -n\leq\nu<0
  \end{array}\right.  \ ,
  $$
i.e.\ the complex is of the form:
   $$ 0\to K_{n+1,n+1}\to \dots\to K_{1,1}\oplus K_{1,n+1}\to
   K_{0,0}\oplus K_{0,n}\to K_{-1,n-1}\to\cdots\to K_{-n,0}\to 0 $$
  We can explicitly give all determinantal integers in this case:
  $$
  K_{2}=0\iff K_{2,2}=0\iff Q_2=\varnothing \iff S_2\subseteq P , 
  $$
  thus
  $$
  \min S_2>\frac m d \iff s_0+s_1>\frac m d \iff m< (s_0+s_1)d.
  $$
  Similarly $K_{-1}=0 \iff Q_{n-1}=\varnothing \iff S_{n-1}\subseteq P$
  and thus
  $$
  \max S_{n-1}\leq \frac{m+n}{d} \iff m \ge d\sum_{i=2}^n s_i -n .
  $$
  Consequently, a determinantal formula exists iff $d \sum_2^n s_i -n
  < (s_0+s_1)d$, also verified by Theorem~\ref{Hdeter}.
  In this case the integers contained in the interval 
 $$
 \Big(d\sum_{i=2}^n s_i -n-1\ ,\ d(s_0+s_1)\Big)
 $$
 are the only determinantal vectors, also verifying
 Corollary~\ref{DetBoxes}.  Notice that the sum of the two endpoints is
 exactly the critical degree $\rho$.

 In \cite[Cor.4.2,Prop.5.6]{DADi01} it is proved that the
 minimum-dimension determinantal formula is attained at
 $m=\left\lfloor \rho/2\right\rfloor$ and $m=\left\lceil
   \rho/2\right\rceil$, ie.\ the center(s) of this interval.  For an
 illustration see Ex.~\ref{unmixedExamBezout}.
\end{example}

\subsection{Pure formulae}

A determinantal formula is pure if it is of the form $K_{1,a}\to
K_{0,b}$ for $a,b\in[0,n+1]$ with $a>b$. These formulae are either
Sylvester- or B\'ezout-type, named after the matrices for the
resultant of two univariate polynomials.

In the unmixed case both kinds of pure formulae exist exactly when for
all $k\in[1,r]$ it holds that $\min\{l_k,d_k\}=1$ \cite{SZ94,DE03}.
The following theorem extends this characterization to the scaled
case, by showing that only pure Sylvester formulae are possible and
the only data that admit such formulae are univariate and
bivariate-bihomogeneous systems.

\begin{theorem} \label{Pdeter} If $\bm s\neq \boldsymbol 1$ a pure
  Sylvester formula exists if and only if $r\leq 2$ and $\bm l= (1)$ or $\bm
  l=(1,1)$. If $l_1=n=1$ the degree vectors are given by
  $$m= d_1\sum_0^1 s_i-1 \ \text{ and }\ m'=-1  , $$
  whereas if $\bm l=(1,1)$ the vectors are given by
  $$
  \bm m=\left(-1, d_2 \sum_0^2 s_i-1\right)\ \text{ and }\ \bm
  m'=\left(d_1 \sum_0^2 s_i-1, -1\right).
  $$
  Pure B\'ezout determinantal formulae \emph{cannot} exist.
\end{theorem}
Notice the duality $\bm m+\bm m'=\rho$.
\begin{proof}
  It is enough to see that if a pure formula is determinantal the
  following inequalities hold
  $$ n\leq \# \bigcup_{p\neq a,b} S_p  \leq \# \cup_1^r P_k \leq n $$
  which implies that equalities hold.  The inequality on the left
  follows from the fact that every $S_p,\ p\in[0,n+1]$ contains at
  least one distinct integer since the sequence
  $0,s_0,s_0+s_1,\dots,\sum_0^n s_i$ is strictly increasing. For the
  right inequality, note that the vanishing of all $K_{\nu,p}$ with
  $p\neq a,b$ implies $Q_p=\varnothing$ (see Lemma~\ref{KnpNz}). Thus
  $\cup_{p\neq a,b} S_p \subseteq \cup_{k=1}^r P_k$ so the cardinality
  is bounded by $\#\cup_{1}^r P_k\leq \sum_1^r \#P_k\leq \sum_1^r l_k
  = n$.  Consequently $\# \cup_{p\neq a,b} S_p = n $. Suppose $n>2$;
  the fact $\#(S_i\cup S_j)>2$ for $\{i,j\}\ne \{0,1\}$ implies
  $\cup_{p\neq a,b} S_p=S_i\cup S_{j}$ for some $i,j$,
  i.e.\ $\#\{a,b\}=n$, contradiction. Thus $n\leq 2$.

  Take $n=2$. Since $\#(S_0\cup S_{3})=2$, the above condition is
  satisfied for $a=2, b=1$: it is enough to set $\cup_{1}^r P_k=
  S_0\cup S_{3}=\{0, \sum_0^2 s_i\}$, thus the integers of $\cup_{1}^r
  P_k$ are not consecutive, so $r>1$ and $l=(1,1)$.  Similarly, if
  $n=l=1$ two formulae are possible; for $\cup_{1}^r P_k=S_0=\{0\}$
  ($a=2, b=1$) or $\cup_{1}^r P_k=S_2=\{s_0+s_1\}$ ($a=1, b=0$).

  All stated $\bm m$-vectors follow easily in both cases from
  $(m_k+l_k)/d_k = 0$ and $(m_k+l_k)/d_k = \sum_0^n s_i $.  A pure
  B\'ezout determinantal formula comes from $K_{1,n+1}\to K_{0,0}$.
  Now $\cup_k P_k$ contains $S_1\cup\cdots\cup S_n$ hence $\#\cup_k
  P_k > n$.  Thus it cannot exist for $\bm s\neq \boldsymbol 1$.
\end{proof}

All pure formulae above are of Sylvester-type, made explicit in
Section~\ref{MatConstr}.  If $n=1$, both formulae correspond to the
classical Sylvester matrix.

If $\bm s=\boldsymbol 1$ pure determinantal formulae are possible for
arbitrary $n,r$ and a pure formula exists if and only if for all $k$, $l_k=1$ or
$d_k=1$ \cite[Thm.4.5]{DE03}; if a pure Sylvester formula exists for
$a, b=a-1$ then another exists for $a=1, b=0$ \cite[p.~15]{DE03}.
Observe in the proof above that this is not the case if $\bm s\neq
\boldsymbol 1, n=2$, thus the construction of the corresponding
matrices for $a\neq 1$ now becomes important and highly nontrivial, in
contrast to~\cite{DE03}.

\section{Explicit matrix construction} \label{MatConstr}

In this section we provide algorithms for the construction of the
resultant matrix expressed as the matrix of the differential
$\delta_1$ in the natural monomial basis and we clarify all the
different morphisms that may be encountered. 

Before we continue, let us justify the necessity of our matrices,
using $\bm l=\bm d=(1,1)$ and $\bm s=(1,1,2)$, that is, the
system of two bilinear and one biquadratic equation to be examined
in Example~\ref{EXmaple}. It turns out that
a (hybrid) resultant matrix of minimum dimension is of size $4\times
4$. The standard B\'ezout-Dixon construction has size $6 \times 6$ but
its determinant is identically zero,
hence it does not express the resultant of the system.

The matrices constructed are unique up to row and column operations,
reflecting the fact that monomial bases may be considered with a
variety of different orderings. The cases of pure Sylvester or pure
B\'ezout matrix can be seen as a special case of the (generally
hybrid, consisting of several blocks) matrix we construct in this
section.

In order to construct a resultant matrix we must find the matrix of
the linear map $\delta_1: K_1\to K_0$ in some basis, ty\-pi\-cal\-ly
the natural monomial basis, provided that $K_{-1}=0$. In this case we
have a generically surjective map with a ma\-xi\-mal minor divisible
by the sparse resultant. If additionally $K_{2}=0$ then $\dim K_1=\dim
K_0$ and the determinant of the square matrix is equal to the
resultant, i.e.\ the formula is determinantal.  We consider
restrictions $\delta_{a,b}: K_{1,a}\to K_{0,b}$ for any direct summand
$K_{1,a},\ K_{0,b}$ of $K_1,\ K_0$ respectively. Every such
restriction yields a block of the final matrix of size defined by the
corresponding dimensions.  Throughout this section the symbols $a$ and
$b$ will refer to these indices.

\subsection{Sylvester blocks} \label{SMatConstr}

The Sylvester-type formulae we consider generalize the classical
univariate Sylvester matrix and the multigraded Sylvester matrices of
\cite{SZ94} by introducing multiplication matrices with block
structure. Even though these Koszul morphisms are known to correspond
to some Sylvester blocks since~\cite{WeZe} (see Proposition~\ref{sylvMap}
below), the exact interpretation of the morphisms into matrix formulae
had not been made explicit until now. We also rectify the
Sylvester-type matrix presented in~\cite[Sect.7.1]{DE03}.

By \cite[Prop.2.5,Prop.2.6]{WeZe} we have the following
\begin{proposition} {\rm \cite{WeZe}} \label{sylvMap} If $a-1<b$ then
  $\delta_{a,b}=0$. Moreover, if $a-1=b$ then $\delta_{a,b}$ is a
  Sylvester map.
\end{proposition}
If $a=1$ and $b=0$ then every coordinate of $\bm m$ is non-negative
and there are only zero cohomologies involved in $K_{1,1}=\bigoplus_i
H^0(\bm m - s_i\bm d)$ and $K_{0,0}=H^0(\bm m)$. This map is a well
known Sylvester map expressing the multiplication
$(g_0,\dots,g_n)\longmapsto \sum_{i=0}^n g_if_i$.  The entries of the
matrix are indexed by the exponents of the basis monomials of
$\bigoplus_i S(\bm m - s_i\bm d)$ and $S(\bm m)$ as well as the chosen
polynomial $f_i$.
 The entry indexed ${(i,\bm \alpha),\bm \beta}$ can be computed as:
 $$
 \text{coef}\left(f_i, \bm x^{\bm \beta-\bm \alpha}\right) \ \ ,\ \
 i=0,1,\dots,n
 $$
 where $\bm x^{\bm \alpha}$ and $\bm x^{\bm \beta}$ run through the
 corresponding monomial bases. The entry $(i,\bm \alpha),\bm \beta$
 is zero if the support of $f_k$ does not contain ${\bm \beta-\bm
   \alpha}$.
Also, by Serre duality a block $K_{1,n+1}\to K_{0,n}$ corresponds to
the dual of $K_{1,1}\to K_{0,0}$, i.e.\ to the degree vector $\bm \rho
-\bm m$, and yields the same matrix transposed.

The following theorem constructs corresponding Sylvester-type matrix in the ge\-ne\-ral case.

\begin{theorem} \label{SylvMat} The entry of the transposed matrix of
  $\delta_{a,b}: K_{1,a}\to K_{0,a-1} $ in row $(I,\bm \alpha)$ and column
  $(J,\bm \beta)$ is
  $$
  \left\{\begin{array}{l}
      0, \mbox{ if } J\not\subset I,\\
      (-1)^{k+1}\text{\rm coef}\left(f_{i_k}, \bm x^{\bm u}\right),
      \mbox{ if }I\setminus J=\{i_k\},
    \end{array} \right.
  $$
  where $I=\{ i_1<i_2<\cdots<i_{a} \}$ and $J=\{
  j_1<j_2<\cdots<j_{a-1} \}$, $I,J\subseteq\{0,\dots,n\}$.  Moreover,
  $\bm \alpha,\bm \beta\in\NN^n$ run through the exponents of monomial
  bases of $H^{a-1}(\bm m- \bm d\sum_{\theta=1}^a s_{i_\theta})$,
  $H^{a-1}(\bm m- \bm d\sum_{\theta=1}^{a-1} s_{j_\theta})$, and $\bm
  u\in\NN^n$, with $u_t= |\beta_t-\alpha_t|$.
\end{theorem}
\begin{proof} 
  Consider a basis of ${\textstyle\bigwedge^a} V$, $\{
  e_{i_1,i_2,\dots,i_{a}}\, :\, 0\leq i_1<i_2<\cdots<i_{a}\leq n \}$
  and similarly for ${\textstyle\bigwedge^{a-1}} V$, where
  $e_{0},\dots,e_{n}$ is a basis for $V$.  This differential expresses
  a classic Koszul map
  $$
  \partial_a(e_{i_1},\dots,e_{i_{a}})= \sum_{k=0}^n (-1)^{k+1} f_{i_k}
  e_{i_1,\dots, i_{k-1}, i_{k+1},\dots,i_{a}}
  $$
  and by \cite[Prop.2.6]{WeZe}, this is identified as multiplication
  by $f_{i_k}$, when passing to the complex of modules.

  Now fix two sets $I\subseteq J$ with $I\setminus J=\{i_k\}$,
  corresponding to a choice of basis elements $e_{I}$, $e_J$ of the
  exterior algebra; then the part of the Koszul map from $e_{I}$ to
  $e_J$ gives
  $$
  (-1)^{k+1}M(f_{i_k}):H^{a-1}(\bm m- \bm d{\sum_{\theta\in I}}
  s_{\theta}) \to H^{a-1}(\bm m- \bm d{\sum_{\theta\in J}} s_{\theta})
  $$
  This multiplication map is a product of homogeneous multiplication
  operators in the symmetric power basis. This includes operators between
  negative symmetric powers, where multiplication is expressed by
  applying the element of the dual space to $f_{i_k}$.

  To see this, consider basis elements $\bm w^{\bm \alpha}$, $\bm
  w^{\bm \beta}$ that index a row and column resp. of the matrix of
  $M(f_{i_k})$. Here the part $\bm w_k$ of $\bm w$ associated with the
  $k$-th variable group is either $\bm x_k^{\bm \alpha_k}$ or a dual
  element indexed by $\bm \alpha_k$.  We identify dual elements with
  the negative symmetric powers, thus this can be thought as
  $\bm x_k^{-\bm \alpha_k}$. This defines $\tilde{\bm \alpha},\tilde{\bm
    \beta}\in\ZZ^n$; the generalized multihomogeneous multiplication
  by $f_{i_k}$ as in~\cite[p.577]{WeZe} is, in terms of multidegrees,
  incrementing $|\tilde{\bm \alpha}_k|$ by $s_id_k$ to obtain $|\tilde{\bm \beta}_k|$,
  and hence the corresponding matrix has entry $\text{\rm
    coef}\left(f_{i_k}, \bm x^{\bm u}\right)$, where $u_t=
  |\beta_t-\alpha_t|$, $t\in[1,n]$.  The absolute value is needed because for
  multiplication in dual spaces, the degrees satisfy $-|\bm \alpha_k| +
  s_id_k=-|\bm \beta_k| \Rightarrow
  s_id_k=|\bm \alpha_k|-|\bm \beta_k|=-(|\bm \beta_k|-|\bm \alpha_k|)$.
\end{proof}

In \cite[Sect.7.1]{DE03}, an example is studied that admits a
Sylvester formula with $a=2, b=1$. The matrix derived by such a
complex is described by Theorem~\ref{SylvMat} above and does not coincide
with the matrix given there. The following example is taken from there
and presents the correct formula.

\begin{example} 
  \label{unmixedExamSylv} Consider the unmixed case $\bm l=(1,1),\
  \bm d=(1,1)$, as in \cite[Sect.7.1]{DE03}. This is a system of three
  bi-linear forms in two affine variables. The vector $\bm m=(2,-1)$
  gives $K_1=K_{1,2}=H^1(0,-3)^{\binom{3}{2}}$ and
  $K_0=K_{0,1}=H^1(1,-2)^{\binom{3}{1}}$. The Sylvester map
  represented here is
  $$ \delta_1\ :\ (g_0,g_1,g_2) \mapsto (-g_0f_1-g_1f_2, g_0f_0-g_2f_2, g_1f_0+g_2f_1)  $$
  and is similar to the one
  in~\cite{DAnEmi02}. By Theorem~\ref{SylvMat}, it yields the
  following (transposed) matrix, given in $2\times 2$ block format:
  $$
  \left[ \begin{array}{ccc}
      -M(f_1) & M(f_0)  &  \bm 0  \\
      -M(f_2) & \bm  0  &  M(f_0) \\
      \bm 0 & -M(f_2) & M(f_1)
    \end{array}\right]
  $$
  If $g=c_0+c_1x_1+c_2x_2+c_3x_1x_2$ the matrix of the multiplication
  map
  $$M(g)\ :\ \ \ S(0)\otimes S(1)^* \ni w \  \longmapsto\  w g \in S(1)\otimes S(0)^* $$
  in the natural monomial basis is $\left[ \begin{array}{cc} c_2 &
      c_3 \\ c_0 & c_1 \end{array}\right]$ as one can easily verify by
  hand calculations or using procedure \texttt{multmap} of our \Maple\
  package presented in Section~\ref{implement}.
\end{example} 

\subsection{B\'ezout blocks}\label{BMatConstr}

A B\'ezout-type block comes from a map of the form $\delta_{a,b}:
K_{1,a} \to K_{0,b}$ with $a-1>b$. In the case $a=n+1$, $b=0$ this is
a map corresponding to the Bezoutian of the system, whereas in other
cases some B\'ezout-like matrices occur, from square subsystems
obtained by hiding certain variables.

Consider the B\'ezoutian, or Morley form~\cite{jouan}, of
$f_0,\dots,$ $f_n$. This is a polynomial of multidegree $(\bm \rho,\bm \rho)$
in $\mathbb F[\bar{\bm x},\bar{\bm y}]$ and can be decomposed as
$$ 
\Delta := \sum_{u_1=0}^{\rho_1}\cdots\sum_{u_r=0}^{\rho_r}
\Delta_u(\bar x)\cdot \bar y^u
$$
where $\Delta_u(\bar x)\in S$ has $\deg \Delta_u(\bar x)=\bm \rho-\bm u$. Here
$\bar{\bm x}=(\bar{\bm x}_{1},\dots, \bar{\bm x}_{r})$ is the set of homogeneous
variable groups and $\bar{\bm y}=(\bar{\bm y}_1,\dots,\bar{\bm y}_r)$ a set of new
variables with the same cardinalities.

The Bezoutian gives a linear map
$$
{\textstyle\bigwedge^{n+1}} V \rightarrow \bigoplus_{m_k\leq \rho_k}
S(\bm \rho-\bm m)\otimes S(\bm m) .
$$
where the space on the left is the $(n+1)$-th exterior algebra of $V=
S(s_0\bm d)\oplus\cdots\oplus S(s_n\bm d)$ and the direct sum runs
over all vectors $\bm m\in\ZZ^r$ with $m_k\leq \rho_k$ for all $k\in
[1,r]$.
In particular, the graded piece of $\Delta$ in degree $(\bm \rho - \bm m, \bm m )$
in $(\bar x,\bar y)$ is
$$ 
\Delta_{\bm \rho-\bm m,\bm m} := \sum_{u_k = m_k } \Delta_u(\bar{\bm x})\cdot \bar{\bm y}^{\bm u}
$$ 
for all monomials $\bar{\bm y}^{\bm u}$ of degree $\bm m$ and coefficients in
$\mathbb F[\bar x]$ of degree $\bm \rho-\bm m$. It yields a map
$$ 
S(\bm \rho - \bm m)^* \longrightarrow S(\bm m) 
$$
known as the B\'ezoutian in degree $m$ of $f_0,\dots,f_n$.  The
differential of $K_{1,n+1}\to K_{0,0}$ can be chosen to be exactly
this map, since evidently $K_{0,0}=H^0(\bm m)\simeq S(\bm m)$ and
$$
K_{1,n+1} = H^n\left(\bm m-\sum_0^n{s_i}\bm d\right) \simeq S
\left(-\bm m+\sum_0^n{s_i}\bm d + \bm l + \boldsymbol 1 \right)^*
$$
according to Serre duality (see Section~\ref{complex}). Thus,
substituting the critical degree vector, we get $K_{1,n+1}=S(\bm \rho-\bm m)^*$.

The polynomial $\Delta$ defined above has $n+r$ homogeneous variables
and its homogeneous parts can be computed using a determinant
construction in~\cite{AwaChkGoz05}, which we adopt here.
We recursively consider, for $k=1,\ldots ,r$ the uniquely defined
polynomials $f_{i,j}^{(1)}$, where $0\leq j\leq l_{k}$, as follows:
\begin{equation}
f_{i} =x_{1,0}f_{i,0}^{\left( 1\right) }+\ldots
	+x_{1,l_{1}}f_{i,l_{1}}^{\left( 1\right) }, \quad 
f_{i,j}^{\left( 1\right) } \in \mathbb{F}\left[ x_{1,j},\ldots ,x_{1,l_{1}}\right]
	\left[ \overline{x}_{2},\ldots ,\overline{x}_{r}\right] ,
\label{Dcompo-f}
\end{equation}
for all $i=1,\ldots ,n.$ To define $f_{i,j}^{(k)},$ for $2\leq
k\leq r$ and $0\leq j\leq l_{j},$ we decompose $f_{i,l_{k-1}}^{(k-1)}$
as in~(\ref{Dcompo-f}) with respect to the group $\bm x_{j}$: 
\[
\begin{array}{ll}
f_{i,l_{k-1}}^{(k-1)} & =x_{k,1}f_{i,1}^{(k)}+\ldots
+x_{k,l_{k}}f_{i,l_{k}}^{(k)} \\ 
f_{i,j}^{(k)} & \in \mathbb{F}\left[ x_{1,l_{1}},\ldots ,x_{k-1,l_{k-1}}\right] 
\left[ x_{k,j},\ldots ,x_{k,l_{k}}\right] \left[ \overline{x}
_{k+1},\ldots ,\overline{x}_{r}\right] .
\end{array}
\]
Overall we obtain a decomposition
$$f_{i}=\sum\limits_{j=0}^{l_{1}-1}x_{1,j}f_{i,j}^{\left( 1\right)
}+x_{1,l_{1}}\sum\limits_{j=0}^{l_{2}-1}f_{i,j}^{\left( 2\right) }+\cdots
+\prod\limits_{t=1}^{k-1}x_{t,l_{t}}\sum%
\limits_{j=1}^{l_{k}-1}x_{k,j}f_{i,j}^{(k)}+\cdots
$$
$$
\cdots+\prod\limits_{t=1}^{r-1}x_{t,l_{t}}\sum%
\limits_{j=1}^{l_{r}-1}x_{r,j}f_{i,j}^{(r)}+\prod%
\limits_{t=1}^{r}x_{t,l_{t}} f_{i,l_{r}}^{(r)}
 f_{i,l_r} \in \mathbb{F}[x_{1,l_1},\dots,x_{1,l_r}]  
$$
of the polynomial $f_{i}$, for all
$i=1,\ldots ,n.$ 
The order of the variable groups, from left to right, corresponds to choosing the
permutation $\pi=\id$. The determinant of size $(n+1)\times(n+1)$ given by
\[
{\cal D=}\left| 
\begin{array}{lllllllllll}
f_{0,0}^{\left( 1\right) } & \ldots & f_{0,l_{1}-1}^{\left( 1\right) } & 
 \ldots & f_{0,0}^{(k)} & \ldots & f_{0,l_{k}-1}^{(k)} & \ldots & 
f_{0,0}^{\left( r\right) } & \ldots & f_{0,l_{r}}^{\left( r\right) } \\ 
\vdots &  & \vdots &  & \vdots &  & \vdots &  & \vdots &  & \vdots \\ 
f_{i,0}^{\left( 1\right) } & \ldots & f_{i,l_{1}-1}^{\left( 1\right) } & \ldots
& f_{i,0}^{(k)} & \ldots & f_{i,l_{k}-1}^{(k)} & \ldots & f_{i,0}^{\left(
r\right) } & \ldots & f_{i,l_{r}}^{\left( r\right) } \\ 
\vdots &  & \vdots &  & \vdots &  & \vdots &  & \vdots &  & \vdots \\ 
f_{n,0}^{\left( 1\right) } & \ldots & f_{n,l_{1}-1}^{\left( 1\right) } & \ldots
& f_{n,0}^{(k)} & \ldots & f_{n,l_{k}-1}^{(k)} & \ldots & f_{n,0}^{\left(
r\right) } & \ldots & f_{n,l_{r}}^{\left( r\right) }
\end{array}
\right|  ,
\]
is equal to $\Delta_{\bm \rho-\bm m,\bm m}$, in our setting, as we have the following:
\begin{theorem}{\rm \cite{AwaChkGoz05}}
  The determinant $\cal D$ is an inertia form of degree $\rho_k-m_k$
  with respect to the variable group $\bm x_k$, $k=1,\dots,r$.
\end{theorem}

Let us show a more simple construction of some part $\Delta_{\bm
  \rho-\bm m,\bm m}$ using an affine B\'ezoutian.
Let $\bm x_k=(x_{k,1},\dots, x_{k,l_k})$ the (dehomogenized) $k$-th
variable group, and $\bm y_k=( y_{k,1},\dots, y_{k,l_k})$.  As a result
the totality of variables is $\bm x=(\bm x_1,\dots,\bm x_r)$ and $\bm
y=(\bm y_1,\dots,\bm y_r)$.

We set $\bm w_{t},\ t=1,\dots n-1$ the conjunction of the first $t$
variables of $\bm y$ and the last $n-t$ variables of $\bm x$.

If $a=n+1,b=0$ the affine B\'ezoutian construction follows from the
expansion of
$$ 
\left|\begin{array}{ccc}
    f_0(\bm x) & f_0(\bm w_1) \ \cdots \ f_0(\bm w_{n-1}) & f_0(\bm y)\\
    \vdots& \vdots \ \ \ \ \ \ \ \ \  \ \ \ \ \  \vdots & \vdots \\
    f_n(\bm x) & f_n(\bm w_1) \ \cdots \ f_n(\bm w_{n-1}) & f_n(\bm y)
\end{array}
\right| / \prod_{k=1}^r\prod_{j=1}^{l_k} (x_{kj}-y_{kj})
$$
as a polynomial in $\mathbb F[y]$ with coefficients in $\mathbb F[x]$.
Hence the entry indexed $\bm \alpha, \bm \beta$ of the B\'ezoutian in some
degree can be computed as the coefficient of $\bm x^{\alpha}\bm
y^{\bm \beta}$ of this polynomial.

We propose generalizations of this construction for arbitrary $a,b$
that are called \emph{partial Bezoutians}, as in~\cite{DE03}.  
It is clear that $a-1=q(z_1)$ and $b=q(z_2)$, for $z_1\in Q_{a}$ and
$z_2\in Q_b$. The difference $a-b-1= \sum_{\theta=1}^{t} l_{k_\theta}$
where $k_1,\dots,k_t$ is a subsequence of $[1,r]$, since if $P_k < b $
then $P_k < a$ thus
$$
q(a)-q(b)=\sum_{P_k < a}l_k - \sum_{P_k < b}l_k= \sum_{ b < P_k <
  a}l_k .
$$

These indices suggest the variable groups that we should substitute in
the partial Bezoutian.  Note that in the case of Bezoutian blocks, it
holds $a-b-1>0$ thus some substitutions will actually take place.  Let
$ i_1,\dots,i_{a-b}$ be a subsequence of $[0,n]$.  We can define a
partial Bezoutian polynomial with respect to
$f_{i_1},\dots,f_{i_{a-b}}$ and $\bm y_{k_1},\dots,\bm y_{k_t}$ as
\begin{equation}\label{pBez}
  \left|\begin{array}{ccccc}
      f_{i_1}(\bm x) & \cdots  & f_{i_1}(\bm w)\\
      \vdots&  & \vdots \\
      f_{i_{a-b}}(\bm x)  & \cdots  & f_{i_{a-b}}(\bm w)
\end{array} \right| /
\prod_{\theta=1}^t\prod_{j=1}^{l_{i_\theta}} (x_{i_\theta, j}-y_{i_\theta, j}) .
\end{equation}
In this Bezoutian, only the indicated $\bm y$-variable substitutions
take place, in successive columns: The variable vector $\bm w$ differs
from $\bm x$ at $\bm x_{k_1},\dots, \bm x_{k_t}$, these have been
substituted gradually with $\bm y_{k_1},\dots, \bm y_{k_t}$. 
Note that $\bm w$ generalizes the vectors $\bm w_t$ defined earlier,
in the sense that the variables of only specific groups are
substituted.
The total number of substituted variables is $a-b-1$, so this is
indeed a B\'ezout type determinant.

For given $a$ and $b$, there exist $\binom{n+1}{a-b}$ partial
Bezoutian polynomials. The columns of the final matrix are indexed by
the $\bm x$-part of their support, and the rows are indexed by the
$\bm y$-part as well as the chosen polynomials
$f_{i_1},\dots,f_{i_{a-b}}$.

\begin{example} \label{unmixedExamBezout} Consider the unmixed data
  $l=2,\ d=2,\ {\bf s}=(1,1,1)$. Determinantal formulae are $m\in[0,3]$,
  which is just $m=0$, $m=1$ and their transposes. Notice how these
  formulae correspond to the decompositions of $\rho=3=3+0=2+1$.  In
  both cases the complex is of block type $K_{1,3}\to K_{0,0}\oplus
  K_{0,2}$. The Sylvester part $K_{1,3}\to K_{0,2}$ can be retrieved
  as in Ex.~\ref{unmixedExamSylv}. For $m=0$ the B\'ezout part is
  $H^2(-6)\simeq S(3)^* \to H^0(0)\simeq S(0)$, whose $5\times 1$
  matrix is in terms of brackets
$$
\left[\begin{array}{ccccc}
\   [142]& [234]+[152]& [235]& [042]& [052]    \
\end{array}\right]^T  .
$$
A bracket $[ijk]$ is defined as
$$
[ijk] := \det\left[\begin{array}{ccc} a_i & a_j & a_k\\ b_i & b_j & b_k\\ c_i & c_j & c_k \end{array}\right],
$$
where $a_i,b_i,c_i$ denote coefficients of $f_0,f_1,f_2$ respectively,
for instance
$f_2=c_{{0}}+c_{{1}}x_{{2}}+c_{{2}}{x_{{2}}}^{2}+c_{{3}}x_{{1}}+c_{{4}}x_{{1}}x_{{2}}+c_{{5}}{x_{{1}}}^{2}$.
Now, for $m=1$ we have B\'ezout part $H^2(-5)\simeq S(2)^*\to H^0(1)\simeq S(1)$,
which yields the $5\times 3$ matrix
$$
\left[\begin{array}{ccccc}
\ [142]       & [152]+[234] & [235]& [042]& [052] \\
\ [152]       & [154]+[235] & [354]& [052]& [054] \\
\ [132]+[042] & [052]+[134] & [135]& [041]+[032] &[051]    
\end{array}\right]^T   .
$$
\end{example}

\section{Implementation}  \label{implement}

We have implemented the search for formulae and construction of the
corresponding resultant matrices in {\sc Maple}.  Our code is based on
that of \cite[Sect.8]{DE03} and extends it to the scaled case. We
also introduce new features, including construction of the matrices of
Section~\ref{MatConstr}; hence we deliver a full package for
multihomogeneous resultants, publicly available at
\href{http://www-sop.inria.fr/galaad/amantzaf/soft.html}{\tt
  www-sop.inria.fr/galaad/amantzaf/soft.html}.

Our implementation has three main parts; given data $(\bm l,\bm d,\bm
s)$ it discovers all possible determinantal formula; this part had
been implemented for the unmixed case in \cite{DE03}. Moreover, for a
specific $\bm m-$vector the corresponding resultant complex is
computed and saved in memory in an efficient representation. As a
final step the results of Section~\ref{MatConstr} are being used to
output the resultant matrix coming from this complex. The main
routines of our software are illustrated in Table~\ref{Tmaple}.

\begin{table} \centering
\begin{tabular}{|l|l|}
\hline
routine  & function\\
\hline \hline
{\tt Makesystem} & output polynomials of type  $(\bm l,\bm d,\bm s)$\\
\hline
{\tt mBezout} & compute the m-B\'ezout bound\\
\hline
{\tt allDetVecs} & enumerate all determinantal $\bm m-$vectors \\
\hline
{\tt detboxes} & output the vector boxes of Cor.~\ref{DetBoxes} \\
\hline 
{\tt findSyl} & output Sylvester type vectors (unmixed case) \\
\hline
{\tt findBez} & find all pure B\'ezout-type vectors. \\
\hline 
{\tt MakeComplex} & compute the complex of an $\bm m-$vector \\
\hline
{\tt printBlocks} & print complex as $\oplus_{a} K_{1,a}\to \oplus_{b} K_{0,b} $\\
\hline
{\tt printCohs} & print complex as $\oplus H^q(\bm u)\to \oplus H^q(\bm v) $\\
\hline
{\tt multmap} & construct matrix $M(f_i): S(\bm u)\to S(\bm v)$ \\
\hline
{\tt Sylvmat} & construct Sylv. matrix $ K_{1,p}\to K_{0,p-1}$ \\
\hline
{\tt Bezoutmat} & construct B\'ezout matrix $ K_{1,a}\to K_{0,b}$ \\ 
\hline
{\tt makeMatrix} & construct matrix $K_1\to K_0$\\
\hline
\end{tabular}
\caption{The main routines of our software\label{Tmaple}.}
\end{table}

The computation of all the $\bm m-$vectors can be done by searching
the box defined in Theorem~\ref{mBounds} and using the filter in
Lemma~\ref{nessCond}.  For every candidate, we check whether the terms
$K_2$ and $K_{-1}$ vanish to decide if it is determinantal.

For a vector $\bm m$, the resultant complex can be computed in an
efficient data structure that captures its combinatorial information
and allows us to compute the corresponding matrix.  More specifically,
a nonzero cohomology summand $K_{\nu,p}$ is represented as a list of
pairs $(c_q,e_p)$ where $c_q=\{k_1,\dots,k_t\}\subseteq [1,r]$ such
that $q=\sum_{i=1}^t l_{k_i}=p-\nu$ and $e_p\subseteq [0,n]$ with
$\#e_p=p$ denotes a collection of polynomials (or a basis element in
the exterior algebra). Furthermore, a term $K_\nu$ is a list of
$K_{\nu,p}$'s and a complex a list of terms $K_\nu$.

The construction takes place block by block. We iterate over all
morphisms $\delta_{a,b}$ and after identifying each of them the
corresponding routine constructs a Sylvester or B\'ezout block. Note
that these morphisms are not contained in the representation of the
complex, since they can be retrieved from the terms $K_{1,a}$ and
$K_{0,b}$.

\begin{example}\label{EXmaple}
  We show how our results apply to a concrete example and demonstrate
  the use of the \Maple\ package on it. The system we consider admits
  a standard B\'ezout-Dixon construction of size $6 \times 6$. But its
  determinant is identically zero, due to the sparsity of the
  supports, hence it neither expresses the multihomogeneous resultant,
  nor provides any information on the roots. Instead our method
  constructs a non-singular $4 \times 4$ hybrid matrix.

Let $\bm l=\bm d=(1,1)$ and $\bm s=(1,1,2)$.\\

{\tt
\noindent > read mhomo-scaled.mpl:
\\
\noindent > l:=vector([1,1]): d:=l: s:= vector([1,1,2]):

\noindent > f:= Makesystem(l,d,s);}
\begin{align*}
f_{{0}}&= a_{{0}}+a_{{1}}x_{{1}}+a_{{2}}x_{{2}}+a_{{3}}x_{{1}}x_{{2}}\\
f_{{1}}&= b_{{0}}+b_{{1}}x_{{1}}+b_{{2}}x_{{2}}+b_{{3}}x_{{1}}x_{{2}}\\
f_{{2}}&= c_{{0}}+c_{{1}}x_{{1}}+c_{{2}}x_{{2}}+c_{{3}}x_{{1}}x_{{2}}+
c_{{4}}{x_{{1}}}^{2}+c_{{5}}{x_{{1}}}^{2}x_{{2}}+\\
& \ \ \ \ + c_{{6}}{x_{{2}}}^{2}+c_{{7}}x_{{1}}{x_{{2}}}^{2}+c_{{8}}{x_{{1}}}^{2}{x_{{2}}}^{2}
\end{align*}
We check that this data is determinantal, using Theorem~\ref{Hdeter}:

{\tt
\noindent > has\_deter( l, d, s);

\centerline{$true$}
}

Below we apply a search for all possible determinantal vectors, by examining
all vectors in the boxes of Corollary~\ref{DetBoxes}. The condition used
here is that the dimension of $K_2$ and $K_{-1}$ is zero, which is
both necessary and sufficient.

{\tt
\noindent > allDetVecs( l, d, s) ;

  \centerline{$[[2, 0, 4], [0, 2, 4], [3, 0, 6], [2, 1, 6], [2, -1, 6], [1, 2, 6], [1, 1, 6],$}
  \centerline{$[1, 0, 6], [0, 3, 6], [0, 1, 6], [-1, 2, 6], [3, 1, 8], [1, 3, 8], [1, -1, 8],$}
  \centerline{$ [-1, 1, 8], [3, -1, 10], [-1, 3, 10]]$}
}

The vectors are listed with matrix dimension as third coordinate.
The search returned 17 vectors; the fact that the number of vectors is
odd reveals that there exists a self-dual vector. The critical degree
is $\bm \rho=(2,2)$, thus $\bm m=(1,1)$ yields the self-dual
formula. Since the remaining 16 vectors come in dual pairs, we only
mention one formula for each pair; finally, the first 3 formulae
listed have a symmetric formula, due to the symmetries present to our
data, so it suffices to list $6$ distinct formulae.

Using Theorem~\ref{Hdeter} we can compute directly determinantal boxes:

{\tt
\noindent > detboxes( l, d, s) ;

\centerline{$[[-1, 1], [1, 3]], [[1, 3], [-1, 1]]$}
}

Note that the determinantal vectors are exactly the vectors in these
boxes.  These intersect at $\bm m=(1,1)$ which yields the self-dual
formula.  In this example minimum dimension formulae correspond to the
centers of the intervals, at $\bm m=(2,0)$ and $\bm m=(0,2)$ as noted
in Conj.~\ref{concenters}.

A pure Sylvester matrix comes from the vector

{\tt
\noindent > m:= vector([d[1]*convert(op(s),`+`)-1, -1]);

\centerline{$\bm m=(3,-1)$}
}

We compute the complex:

{\tt
\noindent > K:= makeComplex(l,d,s,m):
\\ 
\noindent > printBlocks(K); printCohs(K);

\centerline{$K_{1,2} \to  K_{0,1}$}

\centerline{$H^1(1, -3)\oplus H^1(0, -4)^2 \to   H^1(2, -2)^2\oplus H^1(1, -3)$}
}

The dual vector $(-1,3)$ yields the same matrix transposed.  The block
type of the matrix is deduced by the first command, whereas {\tt
  printCohs} returns the full description of the complex.  The
dimension is given by the multihomogeneous B\'ezout bound, see Lemma~\ref{mBezout}, which is equal to:

{\tt
\noindent > mbezout( l, d, s) ;

\centerline{$10$}
}

It corresponds to a ``twisted'' Sylvester matrix:

{\tt
\noindent > makematrix(l,d,s,m); 
}

{\small $$
\left[ \begin {array}{cccccccccc} -b_{{1}}&-b_{{3}}&0&a_{{1}}&a_{{3}}
&0&0&0&0&0\\\noalign{\smallskip}-b_{{0}}&-b_{{2}}&0&a_{{0}}&a_{{2}}&0&0&0
&0&0\\\noalign{\smallskip}0&-b_{{1}}&-b_{{3}}&0&a_{{1}}&a_{{3}}&0&0&0&0
\\\noalign{\smallskip}0&-b_{{0}}&-b_{{2}}&0&a_{{0}}&a_{{2}}&0&0&0&0
\\\noalign{\smallskip}-c_{{4}}&-c_{{5}}&-c_{{8}}&0&0&0&a_{{1}}&0&a_{{3}}
&0\\\noalign{\smallskip}-c_{{1}}&-c_{{3}}&-c_{{7}}&0&0&0&a_{{0}}&a_{{1}}
&a_{{2}}&a_{{3}}\\\noalign{\smallskip}-c_{{0}}&-c_{{2}}&-c_{{6}}&0&0&0&0
&a_{{0}}&0&a_{{2}}\\\noalign{\smallskip}0&0&0&-c_{{4}}&-c_{{5}}&-c_{{8}}
&b_{{1}}&0&b_{{3}}&0\\\noalign{\smallskip}0&0&0&-c_{{1}}&-c_{{3}}&-c_{{7
}}&b_{{0}}&b_{{1}}&b_{{2}}&b_{{3}}\\\noalign{\smallskip}0&0&0&-c_{{0}}&-
c_{{2}}&-c_{{6}}&0&b_{{0}}&0&b_{{2}}\end {array} \right] $$
} 

The rest of the matrices are presented in block format; the same
notation is used for both the map and its matrix. The
dimension of these maps depend on $\bm m$, which we omit to write.  Also,
$B(x_k)$ stands for the partial B\'ezoutian with respect to variables $\bm x_k$.

For $\bm m=(3,1)$ we get $K_{1, 1}\oplus K_{1, 2} \to K_{0, 0}$, or
$$H^0(2, 0)^2\oplus H^1(0, -2)^2 \to  H^0(3,1)$$
$$
\left[\begin{array}{c}
M(f_0)\\
M(f_1)\\
B(x_2)
\end{array}\right]
$$
Symmetric is $\bm m=(1,3)$.

For $\bm m=(3,0)$,  $K_{1,2}\to K_{0,0} \oplus K_{0,1}$:
$$ H^1(1, -2) \oplus H^1(0,-3)^2  \to H^0(3,0) \oplus H^1(1,-2)^2$$
$$
\left[\begin{array}{c|c}
\begin{array}{c}
0\\ M(f_0) \\ \ -M(f_1)
\end{array}
& B(x_2)
\end{array}\right]
$$
Symmetric is $\bm m=(0,3)$.

For $\bm m=(2,1)$, we compute $K_{1, 1} \oplus K_{1, 3}\to K_{0,0}$, or
$$ H^1(1,0)^2 \oplus H^2 (-2, -3) \to H^0(2,1)  $$
$$
\left[\begin{array}{c|c}
\begin{array}{c}
M(f_1)\\ M(f_2)
\end{array}
& \Delta_{(0,1),(2,1)}
\end{array}\right]
$$
Symmetric is $\bm m=(1,2)$.

If $\bm m=(1,1)$, $K_{1, 1} \oplus K_{1, 3}\to K_{0, 0} \oplus K_{0, 2}$, yielding
$$ H^0(0,0)^2 \oplus H^2 (-3, -3)  \to  H^0 (1, 1) \oplus H^2 (-2, -2)^2$$
$$
\left[\begin{array}{cc}
\begin{array}{c}
f_0
\\ f_1
\end{array}
& 0
\\
\Delta_{(1,1),(1,1)}
&
\begin{array}{cc}
M(f_0) & -M(f_1)
\end{array}
\end{array}\right]
$$
We write here $f_i$ instead of $M(f_i)$, since this matrix
is just the $1\times 4$ vector of coefficients of $f_i$.

For $\bm m=(2,0)$, we get $K_{1,2}\oplus K_{1,3} \to K_{0,0} \oplus K_{0,1}$, or
$$ H^1 (0, -2) \oplus H^2 (-2, -4) \to H^0(2, 0) \oplus H^1(0, -2) $$
$$
\left[\begin{array}{cc}
B(x_2) & 0\\
\Delta_{(2,0),(0,2)} & B(x_1)
\end{array}\right]
$$
This is the minimum dimension determinantal complex, yielding a
$4\times 4$ matrix.
\end{example}

\begin{ack}\vspace{-.2cm}
We thank Laurent Bus\'e for his help with Ex.~\ref{unmixedExamSylv}.
Both authors are partly supported by the Marie-Curie IT Network SAGA,
[FP7/2007-2013] grant agreement PITN-GA-2008-214584.  
Part of this work was completed by the second author in fulfillment of
the M.Sc.\ degree in the Department of Informatics and
Telecommunications of the University of Athens.
\end{ack}

\end{document}